\documentclass{article}
   
\newif\ifshorter\shortertrue
 
\usepackage{graphicx}
\usepackage{amsmath,amsfonts,amssymb}
\usepackage{amsthm}
\usepackage{algorithmic,multirow}

\usepackage[ruled,vlined,linesnumbered]{algorithm2e}
\usepackage{tikz}
\usepackage{url}
\usepackage{enumitem}
\usepackage{lhelp}

\def\R {\ensuremath{\mathbb{R}}}

\newtheorem{definition}{Definition}
\newtheorem{theorem}{Theorem}
\newtheorem{corollary}[theorem]{Corollary}
\newtheorem{lemma}[theorem]{Lemma}
\newtheorem{remark}[theorem]{Remark}

\usepackage{authblk}
\usepackage{rotating}

\begin{document}

\title{Cylindrical Algebraic Decompositions \\ for Boolean Combinations}

\author[*]{Russell Bradford}
\author[*]{James H. Davenport}
\author[*]{Matthew England}
\author[**]{Scott McCallum}
\author[*]{David Wilson}
\affil[*]{University of Bath}
\affil[**]{Macquarie University}

\date{Email: \texttt{ \{R.J.Bradford, J.H.Davenport, M.England, D.J.Wilson\}@bath.ac.uk}, {\tt Scott.McCallum@mq.edu.au} }

\maketitle

\begin{abstract} 
This article makes the key observation that when using cylindrical algebraic decomposition  (CAD) to solve a problem with respect to a set of polynomials, it is not always the signs of those polynomials that are of paramount importance but rather the truth values of certain quantifier free formulae involving them. This motivates our definition of a Truth Table Invariant CAD (TTICAD).
We generalise the theory of equational constraints to design an algorithm which will efficiently construct a TTICAD for a wide class of problems, producing stronger results than when using equational constraints alone. The algorithm is implemented fully in {\sc Maple} and we present promising results from experimentation.
\end{abstract}

\section{Introduction}
\label{sec:Intro}

Cylindrical algebraic decompositions (CADs) are a key tool in real algebraic geometry, both for their original motivation, solving quantifier elimination problems, but also for use in many other applications ranging from robot motion planning \cite[etc.]{SchwartzSharir1983b} to programming with complex functions \cite[etc.]{DBEW12}.  Traditionally CADs are produced sign-invariant to a given set of polynomials, (the signs of the polynomials do not vary on the cells of the decomposition).  However, this gives far more information than required for most problems.  The idea of a truth invariant CAD (the truth of a formula does not vary on each cell) was defined in \cite{Brown98} for use in simplifying CADs.  The key contribution of this paper is an approach to construct CADs which are truth invariant without having to first build a sign-invariant CAD.  Actually, we directly build CADs which are truth table invariant, (the truth values of various quantifier free formulae do not vary).  

We present an algorithm to efficiently produce TTICADs for a wide class of problems, utilising the theory of equational constraints \cite{McCallum1999a}.  The algorithm goes further than equational constraints by allowing the creation of smaller CADs in a wider variety of cases; for example disjunctive normal form where each individual conjunction has an equational constraint but no single explicit equational constraint is present for the formula.  The problem of decomposing complex space according to a set of branch cuts for the purpose of algebraic simplification (\cite[etc.]{Phisanbutetal2010a}) is of this case.  

\subsection{Background on CAD}

We briefly remind the reader about the theory of CAD, first proposed by Collins in \cite{Collins1975}.
\begin{definition}
A {\em Tarski formula\/} $F(x_1,\ldots,x_n)$ is a Bool\-ean combination ($\land,\lor,\neg$) of statements about the signs, ($=0,>0,<0$, but therefore $\ne0,\ge0,\le0$ as well), of certain integral polynomials $f_i(x_1,\ldots,x_n)$.  We use {\em QFF} to denote a quantifier free Tarski formula.
\end{definition}
CAD was developed as a tool for the problem of quantifier elimination over the reals: given a quantified Tarski formula
\begin{equation}\label{eq:QT}
Q_{k+1}x_{k+1}\ldots Q_nx_n F(x_1,\ldots,x_n)
\end{equation}
(where $Q_i\in\{\forall,\exists\}$ and $F$ is a QFF), produce an equivalent QFF $\psi(x_1,\ldots,x_k)$. Collins proposed to decompose $\R^n$ cylindrically such that each cell was sign-invariant for all $f_i$ occurring in $F$. Then $\psi$ would be the disjunction of the defining formulae of those cells $c_i$ in $\R^k$ such that (\ref{eq:QT}) was true over the whole of $c_i$, which is the same as saying that (\ref{eq:QT}) is true at any one ``sample point'' of $c_i$.
\par
Collins' algorithm has two phases.  The first, \textit{projection}, applies a projection operator repeatedly to a set of polynomials, each time producing another set in one fewer variables.  Together these sets contain the {\em projection polynomials}.
These are then used in the second phase, \textit{lifting}, to build the CAD incrementally.  First $\R$ is decomposed into cells which are points and intervals corresponding to the real roots of the univariate polynomials.  Then $\R^2$ is decomposed by repeating the process over each cell using the bivariate polynomials at a sample point.  The output for each cell consists of {\em sections} (where a polynomial vanishes) and {\em sectors} (the regions between). Together these form a {\em stack} over the cell, and taking the union of these stacks gives the CAD of $\R^2$.  This is repeated until a CAD of $\R^n$ is produced.  
\par
To conclude that a CAD produced in this way is sign-invariant we need delineability.  A polynomial is {\em delineable} in a cell if the portion of its zero set in the cell consists of disjoint sections.  A set of polynomials are {\em delineable} in a cell if each is delineable and the sections of different polynomials in the cell are either identical or disjoint.  The projection operator used must be defined so that over each cell of a sign-invariant CAD for projection polynomials in $r$ variables, the polynomials in $r+1$ variables are delineable.
\par
The output of a CAD algorithm depends on the variable ordering.  We usually work with polynomials in $\mathbb{Z}[x_1,\ldots,x_n]$ with the variables, ${\bf x}$, in ascending order (so we first project with respect to $x_n$ and continue to reach  univariate polynomials in $x_1$).  The \textit{main variable} of a polynomial (${\rm mvar}$) is the greatest variable present with respect to the ordering. 
\par
Major directions of work since 1975 includes the following:
\begin{enumerate}[itemsep=-5pt,topsep=-8pt]
\item Improvements in Collins' main algorithms by \cite[and many others]{McCallum1988}. These have focussed on reducing the projection sets required as 
%by using the stronger concept of an {\em order-invar\-iant\/} decomposition, 
discussed further later. 
\item Complexity theory of CAD \cite{BrownDavenport2007,DavenportHeintz1988}.
\item Partial CAD, introduced in \cite{CollinsHong1991}, where the structure of $F$ is used to lift %less of the decomposition of $\R^k$ to $\R^n$,  provided it is sufficient to deduce $\psi$.  % JHD was $G$, but that's not used anywhere
only when required to deduce $\psi$.
\item The theory of equational constraints, 
\cite{McCallum1999a,McCallum2001,BrownMcCallum2005} discussed in Section \ref{subsec:EC}.  This is related to the previous direction but differs by using more efficient projections.  %Related to this is ``Variant Quantifier Elimination'' \cite{HongSafeyElDin2012a}.
\item CAD via Triangular Decomposition \cite{Chenetal2009d}: a radically different approach for computing a sign-invariant CAD which is used for \textsc{Maple}'s inbuilt CAD command.
\end{enumerate}

\subsection{TTICAD}
\label{sec:Problem}

We define a new type of CAD, the topic of this paper.
\begin{definition}
Let $\Phi = \{ \phi_i\}_{i=1}^t$ be a list of QFFs.
% in $n$ variables. 
We say a cylindrical algebraic decomposition $\mathcal{D}$ is a {\em Truth Table Invariant} CAD for $\Phi$ (TTICAD) if the Boolean value of each $\phi_i$ is constant (either true or false) on each cell of $\mathcal{D}$.
\end{definition}

A full sign-invariant CAD for the set of polynomials occurring in the formulae of $\Phi$ would clearly be a TTICAD.  However, we aim to produce an algorithm that will construct smaller TTICADs for certain $\Phi$. We will achieve this using the theory of equational constraints (first suggested in \cite{Collins1998} with the key theory developed in \cite{McCallum1999a}).
\begin{definition}
Suppose some quantified formula is given:
\begin{equation*}
\phi^* = (Q_{k+1} x_{k+1})\cdots(Q_n x_n) \phi({\bf x}).
\end{equation*}
where the $Q_i$ are quantifiers and $\phi$ is quantifier free.
An equation $f=0$ is called an {\bf equational constraint} of $\phi^*$ if $f=0$ is logically implied by $\phi$ (the quantifier-free part of $\phi^*$). 
Such a constraint may be either explicit or implicit.
\end{definition}
We suppose that we are given a formula list $\Phi$ in which every QFF $\phi_i$ has a designated explicit equational constraint $f_i = 0$.
We will construct TTICADs by generalising McCallum's reduced projection operator for equational constraints (as in \cite{McCallum1999a}) so that we may make use of the equational constraints.

\subsection{Worked Example}
\label{workedexample}

We will provide details for the following worked example.
%To see the efficacy of TTICAD we will give details for the following example. 

\noindent Consider the polynomials:
\begin{align*}
f_1 := x^2+y^2-1 \qquad \qquad \qquad & g_1 := xy - \tfrac{1}{4} \\
f_2 := (x-4)^2+(y-1)^2-1  \quad & g_2 := (x-4)(y-1) - \tfrac{1}{4}
\end{align*}
which are plotted in Figure \ref{fig:workedexample1}.
We wish to solve the following problem: find the regions of $\R{}^2$ where the formula
\begin{equation*}
\Phi:= \left(f_1 = 0 \land g_1 < 0 \right)\lor \left( f_2 = 0 \land g_2 < 0   \right)
\end{equation*}
is true.  Assume that we are using the variable ordering $y \succ x$ (so the 1-dimensional CAD is with respect to $x$).

Both \textsc{Qepcad} \cite{Brown03} and \textsc{Maple} 16 \cite{Chenetal2009d} produce a full sign-invariant CAD for the polynomials with 317 cells.  At first glance it seems that the theory of equational constraints \cite{McCallum1999a,McCallum2001,BrownMcCallum2005} is not applicable here as neither $f_1 = 0$ nor $f_2 = 0$ is logically implied by $\Phi$.  However, while there is no explicit equational constraint we can observe that $f_1f_2 = 0$ is an {\em implicit} constraint of $\Phi$.  Using \textsc{Qepcad} with this declared gives a CAD with 249 cells.  Later, in Section \ref{subsec:workedexample2} we demonstrate how a TTICAD with 105 cells can be produced.

\begin{figure}
\caption{The polynomials from Section~\ref{workedexample}.}\label{fig:workedexample1}
\begin{center}
\includegraphics[scale=0.2]{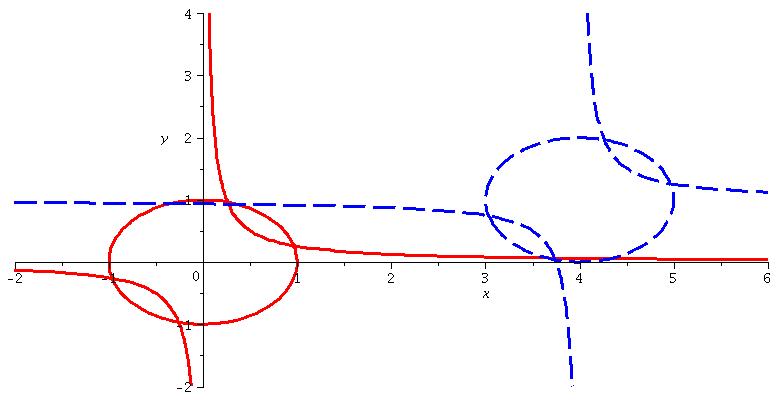}
\end{center}
\vskip-20pt
\end{figure}

\section{Projection Operators}
\label{sec:Project}

\subsection{Equational Constraints}
\label{subsec:EC}

We use two key theorems from McCallum's work on projection and equational constraints.  Both theorems use CADs which are not just sign-invariant but have the stronger property of order-invariance.  A CAD is {\em order-invariant} with respect to a set of polynomials if each polynomial has constant order of vanishing within each cell.

Let $P$ be the McCallum projection operator \cite{McCallum1988}, which produces coefficients, discriminant and cross resultants from a set of polynomials.  We assume the usual trivial simplifications such as removal of constants, exclusion of entries identical to a previous entry (up to constant multiple), and using only the necessary coefficients. 
Recall that a set $A \subset \mathbb{Z}[{\bf x}]$ is an {\em irreducible basis}
if the elements of $A$ are of positive degree in the main variable, irreducible and pairwise relatively prime. 
The main theorem underlying $P$ follows. 

\begin{theorem}[\cite{McCallum1998}]\label{DJW:theorem:SMcCTheorem1}
Let $A$ be an irreducible basis in $\mathbb{Z}[{\bf x}]$ and let $S$ be a connected submanifold of $\mathbb{R}^{n-1}$. Suppose each element of $P(A)$  is order-invariant in $S$.
\noindent Then each element of $A$ either vanishes identically on $S$ or is analytic delineable on $S$, (a slight variant on traditional delineability, see \cite{McCallum1998}). The sections of $A$ not identically vanishing are pairwise disjoint, and each element of $A$ not identically vanishing is order-invariant in such sections.  
\end{theorem}

The main mathematical result underlying the reduction of $P$ in the presence of an equational constraint $f$ is as follows.

\begin{theorem}[\cite{McCallum1999a}]\label{DJW:theorem:SMcCTheorem2}
Let $f({\bf x}), g({\bf x})$ be integral polynomials with positive degree in $x_n$, 
let $r(x_1,\ldots,x_{n-1})$ be their resultant, and suppose $r \neq 0$.
Let $S$ be a connected subset of $\mathbb{R}^{n-1}$
such that $f$ is delineable on $S$ and $r$ is order-invariant in $S$. 
Then $g$ is {\em sign-invariant} in every section of $f$ over $S$.
\end{theorem}

Figure \ref{fig:theorem2} gives a graphical representation of the question answered by Theorem \ref{DJW:theorem:SMcCTheorem2}.  Here we consider polynomials $f(x,y,z)$ and $g(x,y,z)$ of positive degree in $z$ whose resultant $r$ %:= \rm{res}(f,g)$ 
is non-zero, and a connected subset $S \subset \mathbb{R}^2$ in which $r$ is order-invariant.  We further suppose that $f$ is delineable on $S$ (noting that Theorem 1 with $n=3$ and $A = \{f\}$ provides sufficient conditions for this).  We ask whether $g$ is sign-invariant in the sections of $f$ over $S$.  Theorem \ref{DJW:theorem:SMcCTheorem2} answers this question affirmatively:  the real variety of $g$ either aligns with a given section of $f$ exactly (as for the bottom section of $f$ in Figure \ref{fig:theorem2}), or has no intersection with such a section (as for the top). %section of $f$).  
The situation at the middle section of $f$ cannot happen.    
Theorem \ref{DJW:theorem:SMcCTheorem2} thus suggests a reduction of the projection operator $P$ relative to an equational constraint $f = 0$ for the first projection step, as in \cite{McCallum1999a}.

\begin{figure}
\caption{Graphical representation of Theorem \ref{DJW:theorem:SMcCTheorem2}}\label{fig:theorem2}
\begin{center}
\includegraphics[scale=0.45]{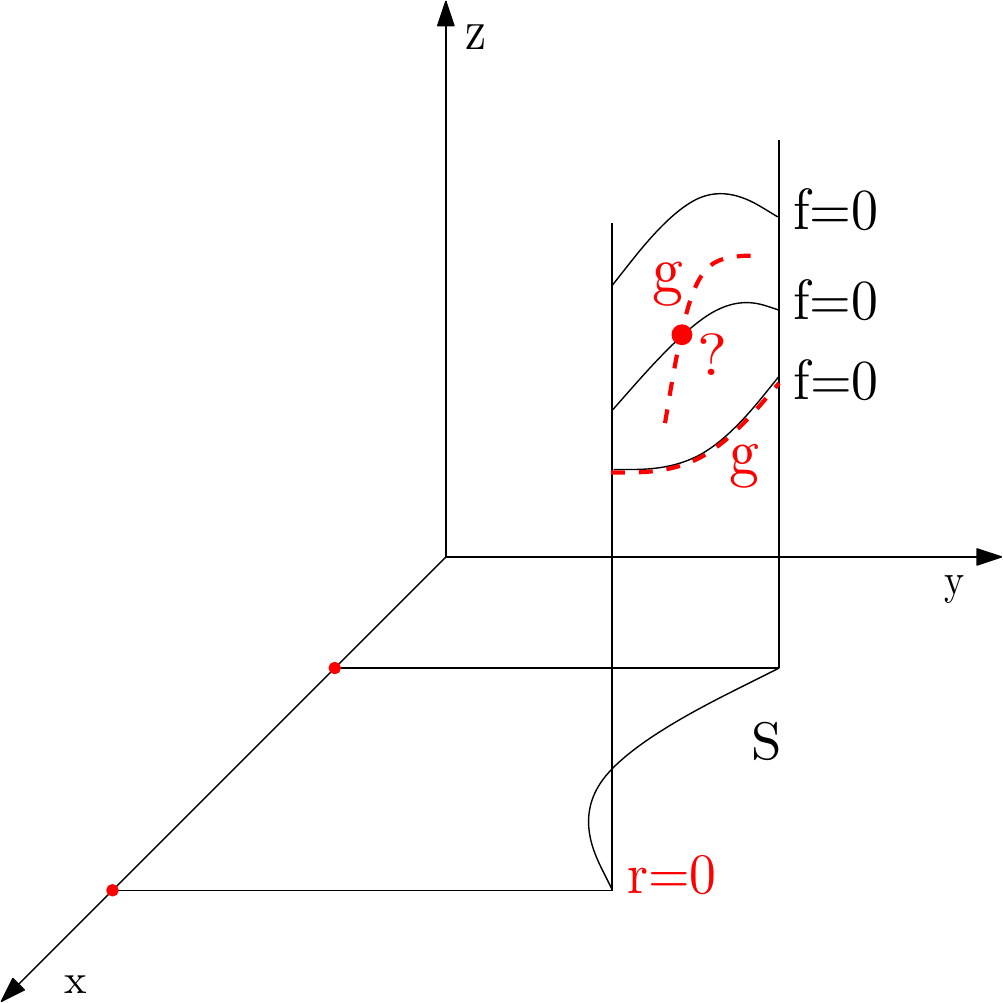}
\end{center}
\vskip-10pt
\end{figure}

\subsection{A Projection Operator for TTICAD}
\label{subsec:ProjOper}

In \cite{McCallum1999a} the central concept is that of the reduced projection
of a set $A$ of integral polynomials relative to a nonempty subset $E$ of $A$ and it is an extension of this which is central here.
For simplicity in \cite{McCallum1999a}, the concept is first defined for the case when $A$ is an irreducible basis and by analogy we start with a similar special case.
Let $\mathcal{A} = \{ A_i\}_{i=1}^t$ be a list of irreducible bases $A_i$
and let $\mathcal{E} = \{ E_i \}_{i=1}^t$ be a list of nonempty subsets
$E_i \subseteq A_i$. 
Put $A = \bigcup_{i=1}^t A_i$ and $E = \bigcup_{i=1}^t E_i$  (we will use the convention of uppercase Roman letters for sets and calligraphic letters for sequences). 
\begin{definition}\label{def:reducedproj}
We define the {\em reduced projection of $\mathcal{A}$ with respect to $\mathcal{E}$}, denoted by $P_{\mathcal{E}}(\mathcal{A})$, as follows:
\begin{equation}
P_{\mathcal{E}}(\mathcal{A}) := \textstyle{\bigcup_{i=1}^t} P_{E_i}(A_i) \cup {\rm Res}^{\times} (\mathcal{E}) \label{eqn:reducedproj}
\end{equation}
where 
\begin{align*}
P_{E_i}(A_i) &= P(E_i) \cup \left\{ {\rm res}_{x_n}(f,g) 
\mid f \in E_i, g \in A_i, g \notin E_i \right\}
\\
{\rm Res}^{\times} (\mathcal{E}) 
&= \{ {\rm res}_{x_n}(f,\hat{f}) \mid \exists i,j :
\, %\\ &\qquad 
f \in E_i, \hat{f} \in E_j, i<j, f \neq \hat{f}  \}
\end{align*}
\end{definition}
In Section \ref{subsec:alg} we build Algorithm \ref{algorithm:TTICADalgorithm} to apply the reduced projection operator for less special input sets by considering contents and irreducible factors of positive degree.
%In Algorithm \ref{algorithm:TTICADalgorithm} below we are able to apply the reduced projection operator for less special input sets by considering contents and irreducible factors of positive degree.

\begin{definition}\label{def:excproj}
The {\em excluded projection polynomials} of $(A_i,E_i)$ 
are those in $P(A)$ but excluded from $P_{\mathcal{E}}(\mathcal{A})$:
\begin{align}
&{\rm ExclP}_{E_i}(A_i) := P(A_i) \setminus P_{E_i}(A_i) \label{eqn:exclprojeqn}\\
& \,\, = \{ {\rm coeffs}(g), {\rm disc}_{x_n}(g), {\rm res}_{x_n}(g,\hat{g}) \mid g, \hat{g} \in A_i \setminus E_i,  g \neq \hat{g} \}.\nonumber
\end{align}
The total set of excluded polynomials, denoted ${\rm ExclP}_{\mathcal{E}}(\mathcal{A})$, consists of all the ${\rm ExclP}_{E_i}(A_i)$, along with the cross resultants of $g_i$ with all of $A_j$ for $i \neq j$. 
\end{definition}

\noindent The following theorem is an analogue of Theorem 2.3 of \cite{McCallum1999a},
and provides the foundation for our algorithm in Section \ref{subsec:alg}.

\begin{theorem}\label{DJW:theorem:SMcC3corrected}
  Let $S$ be a connected submanifold of $\mathbb{R}^{n-1}$. Suppose each element of $P_{\mathcal{E}}(\mathcal{A})$ is order invariant in $S$. Then each $f \in E$ either vanishes identically on $S$ or is analytically delineable on $S$, the sections over $S$ of the $f \in E$ which do not vanish identically are pairwise disjoint, and each element $f \in E$ which does not vanish identically is order-invariant in such sections.
  
{\em Moreover}, for each $i$, with $1 \leq i \leq t$, every $g \in A_i \setminus E_i$ is sign-invariant in each section over $S$ of every $f \in E_i$ which does not vanish identically. 
\end{theorem}

\begin{proof}
The crucial observation is that 
$%\begin{equation*}
P(E) \subseteq P_{\mathcal{E}}(\mathcal{A}).
$ %\end{equation*}
To see this, recall equation \eqref{eqn:reducedproj} and note that we can write
\begin{equation*}
  P(E) = {\textstyle\bigcup_i }P(E_i) \cup {\rm Res}^{\times}(\mathcal{E}).
\end{equation*}

We can therefore apply Theorem \ref{DJW:theorem:SMcCTheorem1} to the set $E$ and obtain the first three conclusions immediately.

There remains the final conclusion to prove. Let $i$ be in the range $1 \leq i \leq t$, let $g \in A_i \setminus E_i$ and let $f \in E_i$; suppose $f$ does not vanish identically on $S$. Now ${\rm res}_{x_n}(f,g) \in P_{\mathcal{E}}(\mathcal{A})$, and so is order-invariant in S by hypothesis. Further, we already concluded that $f$ is delineable. Therefore by Theorem \ref{DJW:theorem:SMcCTheorem2}, $g$ is sign-invariant in each section of $f$ over $S$. 
\end{proof}

In the following section we can use Theorem \ref{DJW:theorem:SMcC3corrected} as the key tool for our implementation of TTICAD, so long as the equational constraint $f$ does not vanish identically on the lower dimensional manifold, $S$.  When working with a polynomial $f$ considered in $r$ variables that vanishes identically at a point  $\alpha \in \R{}^{r-1}$ we say that $f$ is {\em nullified} at $\alpha$.

\begin{remark}
It is clear that the reduced projection $P_{\mathcal{E}}(\mathcal{A})$ will lead to fewer (or the same) projection polynomials than the full projection $P$.  
%One may consider instead using the reduced projection of \cite{McCallum1999a} with a single implicit equational constraint formed as the product of the equational constraint from each QFF.  However, $P_{\mathcal{E}}(\mathcal{A})$ will also lead to fewer (or the same) projection polynomials than this approach since the operator in \cite{McCallum1999a} would consider all resultants res$(e_i,a_j)$ where $e_i \in E, a_j \in A$ while the operator in Definition \ref{def:reducedproj} considers only those with $i=j$.
One may consider instead using the reduced projection $P_E(A)$ of \cite{McCallum1999a}, (with $E=\cup_i E_i$ and $A=\cup_i A_i$ as above).  In the context of Section \ref{sec:Problem} this corresponds to using $\prod_i f_i$ as an implicit equational constraint for a single formula.  Note that $P_{\mathcal{E}}(\mathcal{A})$ also contains fewer polynomials than $P_E(A)$ in general since $P_E(A)$ contains all resultants res$(f,g)$ where $f \in E_i, g \in A_j$ (and $g \notin E$), while $P_{\mathcal{E}}(\mathcal{A})$ contains only those with $i=j$ (and $g \notin E_i$).
\end{remark}

\subsection{Worked Example}
\label{subsec:workedexample2}

In Section \ref{sec:Implementation} we will discuss how to use these results to define an algorithm for TTICAD. First we illustrate the potential savings with our worked example from Section \ref{workedexample}.

In the notation introduced above we have:
\begin{align*}
A_1 := \{f_1,g_1\}, \,\, E_1:=\{ f_1 \}; %\\ 
\,\,
A_2 := \{f_2,g_2\}, \,\, E_2:=\{ f_2 \}.
\end{align*}
We construct the reduced projection sets for each $\phi_i$,
\begin{align*}
  P_{E_1}(A_1) &= \left\{ x^2-1, x^4 - x^2 + \tfrac{1}{16} \right\}, \\
  P_{E_2}(A_2) &= \left\{ x^2 - 8x +15, x^4 -16x^3 + 95x^2-248x + \tfrac{3841}{16} \right\}
\end{align*}
and the cross-resultant set
\begin{equation*}
{\rm Res}^{\times} (\mathcal{E}) = \{{\rm res}_{y}(f_1,f_2)\} = \{ 68x^2 -272x + 285\}.
\end{equation*}

\begin{figure}
\caption{The polynomials from the worked example along with the solutions to the projection sets.}
\label{fig:workedexample4}
\begin{center}
\includegraphics[scale=0.22]{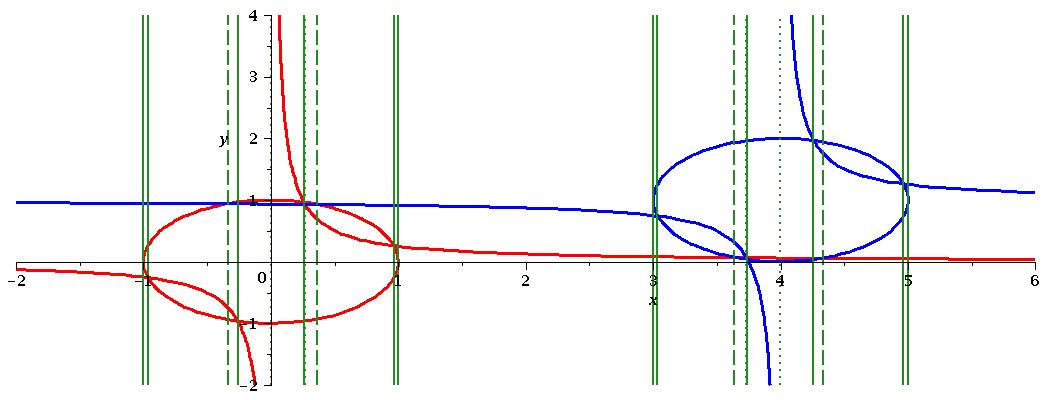}
\end{center}
\vskip-10pt
\end{figure}

$P_{\mathcal{E}}(\mathcal{A})$ is then the union of these three sets.  In Figure \ref{fig:workedexample4} we plot the polynomials (solid curves) and identify the 12 real solutions of $P_{\mathcal{E}}(\mathcal{A})$ (solid vertical lines).  We can see the solutions align with the asymptotes of the $f$s and the important intersections (those of $f_1$ with $g_1$ and $f_2$ with $g_2$).   

%If we were to instead use a projection operator based on an implicit equational constraint $f_1f_2=0$ then we would have an extra 4 solutions (the dashed vertical lines) which align with the intersections of $f_1$ with $g_2$ and $f_2$ with $g_1$.  Finally, if we were to consider $P(\{f_1,g_1,f_2,g_2\})$ then we would gain another 4 solutions (the dotted vertical lines) which align with the intersections of $g_1$ and $g_2$ and the asymptotes of the $g$s.  
If we were to instead use a projection operator based on an implicit equational constraint $f_1f_2=0$ then in the notation above we would construct $P_E(A)$ from $A=\{f_1,f_2,g_1,g_2\}$ and $E=\{f_1,f_2\}$.  This set provides an extra 4 solutions (the dashed vertical lines) which align with the intersections of $f_1$ with $g_2$ and $f_2$ with $g_1$.  
Finally, if we were to consider $P(A)$ then we gain another 4 solutions (the dotted vertical lines) which align with the intersections of $g_1$ and $g_2$ and the asymptotes of the $g$s.  
In Figure \ref{fig:workedexample5} we magnify a region
%the region $0.2 \leq x \leq 0.5$, $0.8 \leq y \leq 1.1$ 
to show explicitly that the point of intersection between $f_1$ and $g_1$ is identified in $P_{\mathcal{E}}(\mathcal{A})$, whereas the intersection points of $g_2$ with both $f_1$ and $g_1$ are ignored.

\begin{figure}
\caption{Magnified region of Figure \ref{fig:workedexample4}}\label{fig:workedexample5}
\begin{center}
\includegraphics[scale=0.25]{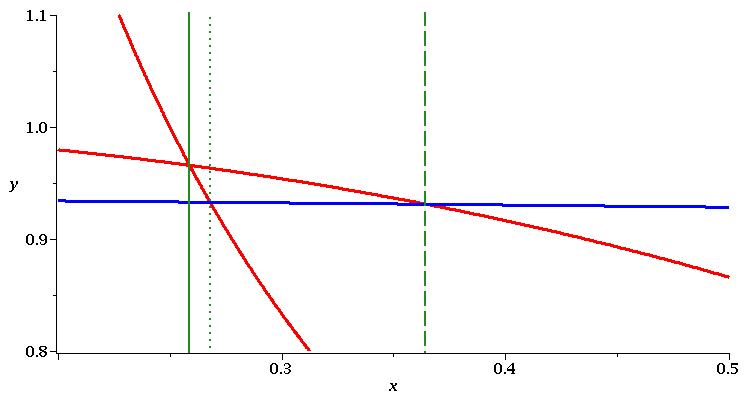}
\end{center}
\vskip-20pt
\end{figure}

Hence the 1-dimensional CAD produced using $P_{\mathcal{E}}(\mathcal{A})$ has 25 cells compared to 33 when using $P_E(A)$ and 41 when using $P(A)$.  However, it is important to note that this reduction is amplified after lifting (using Theorem \ref{DJW:theorem:SMcC3corrected} and and Algorithm \ref{algorithm:TTICADalgorithm}).  The full dimensional TTICAD has 105 cells, the CAD invariant with respect to the implicit equational constraint has 249 cells  and the full sign-invariant CAD has 317.

\section{Implementation}
\label{sec:Implementation}

\subsection{Algorithm Description and Proof} 
\label{subsec:alg}

We describe carefully Algorithm \ref{algorithm:TTICADalgorithm}.  This will create a TTICAD of $\R{}^n$ for a list of QFFs, $\Phi = \{ \phi_i \}_{i=1}^t$, in variables ${\bf x} = x_1 \prec x_2 \prec \cdots \prec x_n$ where each $\phi_i$ has a designated equational constraint $f_i = 0$ of positive degree.  
We use a subalgorithm \texttt{CADW}, fully specified and validated in \cite{McCallum1998}.  The input of {\tt CADW} is: $r$, a positive integer and $A$, a set of $r$-variate integral polynomials. The output is a Boolean $w$ which if true is accompanied by an order-invariant CAD for $A$ (a list of indices $I$ and sample points $S$). 

Let $A_i$ be the set of all polynomials occurring in $\phi_i$,
put $E_i = \{f_i\}$, and let $\mathcal{A}$ and $\mathcal{E}$
be the lists of the $A_i$ and $E_i$, respectively.
Our algorithm effectively defines the reduced projection of $\mathcal{A}$ 
with respect to $\mathcal{E}$ using the special case
of this definition from the previous section. The definition amounts to using 
$ \mathfrak{P} := C \cup P_{\mathcal{F}}(\mathcal{B})$ for $P_{\mathcal{E}}(\mathcal{A})$,
where $C$ is the set of contents of all the elements of all the $A_i$,
$\mathcal{B}$ is the list 
$\{B_i\}_{i=1}^t$, such that $B_i$ is the finest squarefree
basis for the set ${\rm prim}(A_i)$
of primitive parts of elements of $A_i$ which have 
positive degree, and $\mathcal{F}$ is the list
$\{F_i\}_{i=1}^t$, such that $F_i$ is the finest squarefree basis
for ${\rm prim}(E_i)$. (The reader will notice that this notation and the 
definition of $P_{\mathcal{E}}(\mathcal{A})$ is analogous to the work in Section 5 of \cite{McCallum1999a}.)

\begin{algorithm}[h!]\label{algorithm:TTICADalgorithm}
%\DontPrintSemicolon

\SetKwInOut{Input}{Input}\SetKwInOut{Output}{Output}
\Input{A list of quantifier-free formulae $\Phi = \{ \phi_i \}_{i=1}^t$ in variables $x_1,\ldots,x_n$. Each $\phi_i$ has a designated equational constraint $f_i = 0$.
}
\Output{Either 
$\bullet$ $\mathcal{D}:$ A TTICAD of $\R{}^n$ for $\Phi$ (described by lists $I$ and $S$
   of cell indices and sample points, respectively); or \qquad \qquad  {\ } 
$\bullet$~{\bf FAIL}: If $\Phi$ is not well oriented
%If$\mathcal{A}$ is not well-oriented with respect to $\mathcal{E}$ 
(Def. \ref{def:tticadwellorientedness}).
}
\BlankLine
\For{$i = 1 \dots t$}{
Set $E_i \leftarrow \{f_i\}$.
Compute the finest squarefree basis $F_i$ for ${\rm prim}(E_i)$\;
} 
%$\mathcal{E} \leftarrow (E_i)_{i=1}^t$
Set $F \leftarrow \cup_{i=1}^t F_i$\;
\eIf{$n=1$}{
%Set $I \leftarrow$ the empty list and $S \leftarrow$ the empty list\label{alg:step:base1}\;
Isolate in $(I,S)$ the real roots of the product of the polynomials in $F$\label{alg:step:base1}\;
%Construct the indices and sample points of the cells of $\mathcal{D}$, adding them to $I$ and $S$ respectively\;
\Return $I$ and $S$ for $\mathcal{D}$\label{alg:step:base2}\;
}
{
\For{$i = 1 \dots t$}{
Extract the set $A_i$ of polynomials in $\phi_i$ \; 
Compute the set $C_i$ of contents of the elements of $A_i$; 
Compute the set $B_i$, the finest squarefree basis for ${\rm prim}(A_i)$\;
} 
Set $C \leftarrow \cup_{i=1}^t C_i$,
%$\mathcal{A} \leftarrow (A_i)_{i=1}^t$,
%$\mathcal{E} \leftarrow (E_i)_{i=1}^t$,
$\mathcal{B} \leftarrow (B_i)_{i=1}^t$ and 
$\mathcal{F} \leftarrow (F_i)_{i=1}^t$ \;
Construct the projection set: $\mathfrak{P} \leftarrow  C \cup P_{\mathcal{F}}(\mathcal{B})$ \;
%Put $A \leftarrow \cup_{i=1}^t A_i$.

Attempt to construct a lower-dimensional CAD: $w',I',S' \leftarrow {\tt CADW}(n-1,\mathfrak{P})$\label{alg:step:cadw}\;
\If{$w' = false$}{
\Return {\bf FAIL} ($\mathfrak{P}$ not well oriented)\label{alg:step:fail}\;
}
%Otherwise $I'$ and $S'$ specify a $P$-invariant CAD $\mathcal{D}'$ of $\R{}^{n-1}$ such that every cell of $\mathcal{D}'$ is a submanifold of $\R{}^{n-1}$ and each polynomial in $P$ is order-invariant in each cell of $\mathcal{D}'$. 
\label{TTICADn-1dimstep}

%Initialise $I$ and $S$ for stack construction:
%Set $I \leftarrow$ the empty list and $S \leftarrow$ the empty list\label{step:lifting1}\;
$I \leftarrow\emptyset$; $S \leftarrow\emptyset$\label{step:lifting1}\;
\For{each cell $c \in \mathcal{D}'$}{
$L_c \leftarrow \{\}$\;
\For{$i = 1,\ldots t$}{
\eIf{$f_i$ is nullified on $c$}{
\eIf{$\dim(c)>0$}{
\Return {\bf FAIL} \label{notwellor2}  ($\Phi$ not well oriented)\;}{
$L_c \leftarrow L_c \cup B_i$\label{addthebi}\;}
}{
$L_c \leftarrow L_c \cup F_i$\;
}
}
Lift over $c$ using $L_c$: construct cell indices and sample points for the stack over $c$ of the polynomials in $L_c$, adding them to $I$ and $S$\label{step:lifting2}\;
}

%For each cell $c$ of $\mathcal{D}'$, lift over $c$ by using those elements of $F$ which do not vanish identically on $c$ to build a stack over $c$; in particular,
%construct cell indices and sample points for the sections and sectors over $c$ of the polynomials in $F$, adding them to $I$ and $S$, respectively.
%After each cell $c$ has been processed in this way, $I$ and $S$ describe a CAD $\mathcal{D}$ of $\R{}^n$ which has the desired properties. \;
\Return $I$ and $S$ for $\mathcal{D}$\;
\caption{{TTICAD Algorithm}}
}
\end{algorithm}

We shall prove that, provided $\mathcal{A}$ and $\mathcal{E}$ are well-oriented as in Definition \ref{def:tticadwellorientedness}, the output of  Algorithm \ref{algorithm:TTICADalgorithm} is indeed a TTICAD for $\Phi$.
Note that this condition is specialised and new, introduced for this paper.  Its requirement is due to both the use of \texttt{CADW} from \cite{McCallum1998} and the introduction of our new reduced projection operator.

We first recall the more general notion of well-orientedness from \cite{McCallum1998}.  A set $A$ of
$n$-variate polynomials is said to be {\em well oriented} if whenever $n > 1$,
every $f \in {\rm prim}(A)$ is nullified by at most a finite number of points
in $\R^{n-1}$, and (recursively) $P(A)$ is well-oriented.  The Boolean output of \texttt{CADW} is false if the input set was not well-oriented in this sense. 
Now we define our new notion of well-orientedness for the set lists $\mathcal{A}$ and $\mathcal{E}$ defined above, and hence $\Phi$.
\begin{definition}\label{def:tticadwellorientedness}
We say $\mathcal{A}$ is {\em well oriented with respect to} $\mathcal{E}$ (and that $\Phi$ is {\em well oriented}) if whenever $n>1$, every constraint polynomial $f_i$ is nullified by at most a finite number of points in $\R^{n-1}$, and $P_{\mathcal{E}}(\mathcal{A})$ 
(hence $\mathfrak{P}$ in the algorithm) is well-oriented in the sense of \cite{McCallum1998}.
\end{definition}
\begin{theorem}
The output of Algorithm \ref{algorithm:TTICADalgorithm} is as specified.  
\end{theorem}
\begin{proof}
We must show that when $\Phi$ is well-oriented the output is a Truth Table Invariant CAD, (each $\phi_i$ has constant truth value in each cell of $\mathcal{D}$), and \textbf{FAIL} otherwise.

If the input was univariate then it is trivially well-oriented.  The algorithm will construct a CAD $\mathcal{D}$ of $\R^1$ using the roots of the irreducible factors of the constraint polynomials (steps \ref{alg:step:base1} to \ref{alg:step:base2}).
At each 0-cell all the polynomials in each $\phi_i$ trivially have constant signs, and hence every $\phi_i$ has constant truth value.  In each 1-cell no constraint polynomial has a root, so every $\phi_i$ has constant truth value $false$.

Now suppose $n > 1$.  
If $\mathfrak{P}$ is not well-oriented in the sense of \cite{McCallum1998} then \texttt{CADW} returns $w'$ as false.  In this case the input $\Phi$ is not well oriented in the sense of Definition \ref{def:tticadwellorientedness} and Algorithm \ref{algorithm:TTICADalgorithm} correctly returns \textbf{FAIL}.  Otherwise, $\mathfrak{P}$ is well-oriented and at step \ref{alg:step:cadw} we have $w'= true$.  Further, $I'$ and $S'$ specify a CAD, $\mathcal{D}'$, order-invariant with respect to $\mathfrak{P}$.
%(by the correctness of \texttt{CADW}, as proved in \cite{McCallum1998}). 
%Every cell in $\mathcal{D}'$ is a submanifold of $\R^{n-1}$
%and each element of $\mathfrak{P}$ is order-invariant in each such cell.
Let $c$, a submanifold of $\R{}^{n-1}$, be a cell of $\mathcal{D}'$.% and $\alpha$ be its sample point.  

Suppose first that the dimension of $c$ is positive.
If any constraint polynomial $f_i$ vanishes identically on $c$ then $\Phi$ is not well oriented in the sense of Definition \ref{def:tticadwellorientedness} and the algorithm correctly returns \textbf{FAIL} at step \ref{notwellor2}.  
Otherwise, we know that $\Phi$ is certainly well-oriented.  Since no constraint polynomial $f_i$ vanishes then no element of the basis $F$ vanishes identically on $c$ either.  Hence, by Theorem \ref{DJW:theorem:SMcC3corrected}, applied with $\mathcal{A} = \mathcal{B}$ and $\mathcal{E} = \mathcal{F}$, each element of $F$ is delineable on $c$, and the sections over $c$ of the elements of $F$ are pairwise disjoint.
Thus the sections and sectors over $c$ of the elements of $F$ comprise
a stack $\Sigma$ over $c$.  Furthermore, Theorem \ref{DJW:theorem:SMcC3corrected} assures us that, for each $i$, every element of $B_i \setminus F_i$ is sign-invariant in each section over $c$ of every element of $F_i$.

Let $1 \le i \le t$. Consider first a section $\sigma$ of the stack $\Sigma$.
We shall show that $\phi_i$ has constant truth value in $\sigma$.
Now the constraint polynomial $f_i$ is a product of its content ${\rm cont}(f_i)$ 
and some elements of the basis $F_i$.
But ${\rm cont}(f_i)$, an element of $\mathfrak{P}$,
is sign-invariant % (indeed order-invariant)
in the whole cylinder $c \times \R$ %, and hence in particular in $\sigma$.
which includes $\sigma$. Moreover all of the elements of $F_i$ are sign-invariant in $\sigma$, as noted previously. Therefore $f_i$ is sign-invariant in $\sigma$.
If $f_i$ is positive or negative in $\sigma$ then
$\phi_i$ has constant truth value $false$ in $\sigma$.

Suppose that $f_i = 0$ throughout $\sigma$.
It follows that $\sigma$ must be a section of some element of the basis $F_i$.
Let $g \in A_i \setminus E_i$ be a non-constraint polynomial in $A_i$.
Now, by the definition of $B_i$, we see $g$ can be written as 
$g = {\rm cont}(g) h_1^{p_1} \cdots h_k^{p_k}$ where $h_j \in B_i, p_j \in \mathbb{N}$. 
But ${\rm cont}(g)$, in $\mathfrak{P}$,
is sign-invariant %(indeed order-invariant)
in the whole cylinder $c \times \R$ %, and hence in particular in $\sigma$.
including $\sigma$.
Moreover each $h_j$ is sign-invariant in $\sigma$, as noted previously.
Hence $g$ is sign-invariant in $\sigma$.  (Note that in the case where $g$ does not have main variable $x_n$ then $g = {\rm cont}(g)$ and the conclusion still holds).
Since $g$ was an arbitrary element of $A_i \setminus E_i$,
it follows that all polynomials in $A_i$ are sign-invariant in $\sigma$,
and hence that $\phi_i$ has constant truth value in $\sigma$.

Next consider a sector $\sigma$ of the stack $\Sigma$, and notice that
at least one such sector exists.
As observed above, ${\rm cont}(f_i)$ is sign-invariant in $c$,
and $f_i$ does not vanish identically on $c$.
Hence ${\rm cont}(f_i)$ is non-zero throughout $c$. Moreover
each element of the basis $F_i$ is delineable on $c$.
Hence the constraint polynomial $f_i$ is nullified by no point of $c$.
It follows from this that the algorithm does not return \textbf{FAIL}
during the lifting phase.
It follows also that $f_i \neq 0$ throughout $\sigma$.
Therefore $\phi_i$ has constant truth value $false$ in $\sigma$.

It remains to consider the case in which the dimension of $c$ is 0.
In this case the roots of the polynomials in the lifting set $L_c$ constructed
by the algorithm determine a stack $\Sigma$ over $c$.
Each $\phi_i$ trivially has constant truth
value in each section (0-cell) of this stack,
and the same can routinely be shown for each sector (1-cell) of this stack.
\end{proof}

\begin{remark}
When the input to Algorithm \ref{algorithm:TTICADalgorithm} is a single QFF then it produces a CAD which is invariant with respect to the sole equational constraint.  This may be shown using the results of \cite{McCallum1999a} alone.  However, we note that  Algorithm \ref{algorithm:TTICADalgorithm} is actually more efficient in the lifting stage than the modified QEPCAD algorithm discussed in \cite{McCallum1999a} since the lifting set excludes some non-equational constraint input polynomials.
\end{remark} 

Algorithm \ref{algorithm:TTICADalgorithm} and Definition \ref{def:tticadwellorientedness} have been kept conceptually simple to aid readability. However in practice the algorithm may sometimes be unnecessarily cautious.  
In \cite{Brown2005}, several cases where non-well oriented input can still lead to an order-invariant CAD are discussed. Similarly here, we can sometimes allow the nullification of an equational constraint on a positive dimensional cell.%, which motivates the Lemma below.

\begin{lemma}\label{lemma:constpolys}
Let $f_i$ be an equational constraint which vanishes identically on a cell $c \in \mathcal{D}'$ constructed during Algorithm \ref{algorithm:TTICADalgorithm}. If all polynomials in ${\rm ExclP}_{E_i}(A_i)$ are constant on $c$ then any $g \in A_i \setminus E_i$ will be delineable over $c$. 
\end{lemma}

\begin{proof}
Suppose first that $A_i$ and $E_i$ satisfy the simplifying conditions from Section \ref{subsec:ProjOper}. Rearranging \eqref{eqn:exclprojeqn} we see
  \begin{equation*}
    P(A_i) = P_{E_i}(A_i) \cup {\rm ExclP}_{E_i}(A_i).
  \end{equation*}

However, given the conditions of the lemma, this is equivalent (after the removal of constants which do not affect CAD construction) to $P_{E_i}(A_i)$ on $c$.  So here $P(A_i)$ is a subset of $P_{\mathcal{E}}(\mathcal{A})$ and we can conclude by Theorem \ref{DJW:theorem:SMcCTheorem1} that all elements of $A_i$ vanish identically on $c$ or are delineable over $c$. 

In the more general case we can still draw the same conclusion because $P(A_i) = C_i \cup P_{F_i}(B_i) \cup {\rm ExclP}_{F_i}(B_i) \subseteq \mathfrak{P}$.
\end{proof}

Hence we can use Lemma \ref{lemma:constpolys} to safely extend step \ref{addthebi} to also apply in such cases. In particular, we can allow equational constraints $f_i$ which do not have main variable $x_n$ in such cases. We have included this in our implementation discussed in Section \ref{MapleSec}.
In  theory, we may be able to go further and allow step \ref{addthebi} to apply in cases where the polynomials in ${\rm ExclP}_{E_i}(A_i)$ are not necessarily all constant, but have no real roots within the cell $c$. However, identifying such cases would require answering a separate quantifier elimination question, which may not be trivial. 

\subsection{TTICAD via the ResCAD Set}

In Algorithm \ref{algorithm:TTICADalgorithm} the lifting stage (steps \ref{step:lifting1} to \ref{step:lifting2}) varies according to whether an equational constraint is nullified.  When this does not occur there is an alternative implementation of TTICAD which would be simpler to introduce into existing CAD algorithms.  
Define the {\em ResCAD Set} of $\Phi$ as
\begin{equation*}
  \mathcal{R}(\Phi) = E \cup {\textstyle \bigcup_{i=1}^t} \left\{ {\rm res}_{x_n}(f,g) \mid f\in E_i, g \in A_i, g \notin E_i \right\}.
\end{equation*}

\begin{theorem} \label{THEOREM:PROJEQUALSRES}
   Let $\mathcal{A} = ( A_i)_{i=1}^t$ be a list of irreducible bases $A_i$
and let $\mathcal{E} = ( E_i )_{i=1}^t$ be a list of non-empty subsets
$E_i \subseteq A_i$.   
  For the McCallum projection operator $P$, \cite{McCallum1988} we have:
  \begin{equation*}
    {P}(\mathcal{R}(\Phi)) = {P}_{\mathcal{E}}(\mathcal{A}).
  \end{equation*}
\end{theorem}
The proof is straightforward and so omitted here.  
%The proof is given in Appendix \ref{appendix:rescad}.

\begin{corollary}\label{corollary:rescad}
If no $f_i$ is nullified by a point in $\R{}^{n-1}$ then inputting $\mathcal{R}(\Phi)$ into any algorithm which produces a sign-invariant CAD using McCallum's projection operator, will result in the TTICAD for $\Phi$ produced by Algorithm \ref{algorithm:TTICADalgorithm}.
\end{corollary}

Hence Corollary \ref{corollary:rescad} gives us a simple way to compute TTICADs using existing CAD implementations, such as {\sc Qepcad}, but this cannot be applied as widely as Algorithm \ref{algorithm:TTICADalgorithm}.

\subsection{Implementation in Maple} 
\label{MapleSec}

There are various implementations of CAD available but none guarantee order-invariance, required for proving the validity of our TTICAD algorithm.  
%Maple's in-built algorithm \cite{Chenetal2009d} is not based on projection and lifting and it is not obvious how one could efficiently adapt it to give order-invariance, while the external package SyNRAC only uses Collins projection operator for CAD.
Hence we needed to construct our own implementation to obtain experimental results.  We built an implementation of McCallum projection, so that we could reproduce {\tt CADW} and modified the existing stack generation commands in \textsc{Maple} from \cite{Chenetal2009d} so they could be used more widely.  Together these allowed us to fully implement Algorithm \ref{algorithm:TTICADalgorithm}.  The CAD implementation grew to a \textsc{Maple} package \texttt{ProjectionCAD} which gathers together algorithms for producing CADs via projection and lifting to complement the existing CAD commands in \textsc{Maple} which use triangular decomposition, giving the same representation of sample points using regular chains.  For further details (along with free access to the code) see \cite{ProjectionCAD}.

%ResCAD can be implemented by Corollary \ref{corollary:rescad}, using the generic CAD algorithm with the option to use McCallum projection specified, although its output will only be a TTICAD if the equational constraints are not nullified.

\subsection{Formulating a Problem for TTICAD} 
\label{subsection:heuristic}

When formulating a problem for TTICAD there may be choices for the input, such as choosing which equational constraint to designate in a QFF when more than one is present.  Other possibilities include choosing whether conjunctions of formulae should be split into separate QFFs.  Usually it will be preferable to minimise the number of QFFs, but if for example a designated equational constraint has many intersections with another polynomial which could be ignored by using separate QFFs, then the cost of the extra polynomials in the projection set may be outweighed by the complexity of those removed.  
Hence it is worth taking care in how we formulate the TTICAD.  
%Note that we may also have the choice to further refine an input connected by disjunctions, in the case where a single QFF has more than one equational constraint.   
A simple problem of the form
\[
f_1 = 0 \land f_2 = 0 \land g_1 < 0 \land g_2 < 0
\]
has six acceptable choices for the composition of $\Phi$.

We have started exploring heuristics for choosing the best composition. The metric  {\tt sotd} (sum of total degrees) as defined in \cite{Dolzmannetal2004a} may be used to approximate the complexity of polynomials. 
%We first considered
%$
%{\tt sotd}\left( P(A) \setminus \mathfrak{P} \right)
%$
%as a heuristic.  This measures the complexity of the first level of projection polynomials excluded by the TTICAD theory and seems to be fairly well correlated with the number of cells produced by Algorithm \ref{algorithm:TTICADalgorithm}.  However it is not fine enough to separate compositions which can lead to TTICADs with significantly different numbers of cells.
We first considered using \texttt{sotd}$(\mathfrak{P})$ and found that while it was fairly well correlated with the number of cells produced by Algorithm \ref{algorithm:TTICADalgorithm} it was not always fine enough to separate compositions leading to TTICADs with significantly different numbers of cells.  Hence we prefer a stronger heuristic, \texttt{sotd}$(\mathfrak{P} \cup \overline{P}(\mathfrak{P}) ))$ where $\overline{P}$ is the complete set of projection polynomials obtained by repeatedly applying $P$.  

For the problems in Section \ref{sec:Experiment} we used the QFFs imposed by the disjunctions of formulae using this heuristic to choose which equational constraints are designated when there was a choice. For these problems the heuristic computation time was negligible compared to the overall time, but for larger problems this would not be the case.  Work on heuristics is ongoing with a more detailed report available in \cite{BDEW13}. 

%Define $\overline{P}$ to be the complete set of projection polynomials obtained by repeatedly applying the operator $P$. To obtain a finer heuristic we consider
%\[
%{\tt sotd}\left( \overline{P}(A) \setminus \left( \mathfrak{P} \cup \overline{P}%(\mathfrak{P}) \right)  \right).
%\]
%This heuristic is more likely to identify a single `best' composition for a problem and seems to correlate well with the smallest TTICAD.  This heuristic is used to decide the composition of $\Phi$, (when there is a choice), for the problems in Section \ref{sec:Experiment}.

%We note that these heuristics do not take into account the augmentation of $L_c$ in Algorithm \ref{algorithm:TTICADalgorithm} when an $f_i$ is nullified, but it does not appear that there is a simple way to take care of this whilst keeping the heuristic computationally cheap.

\section{Experimental Results}
\label{sec:Experiment}

\subsection{Description of experiments}

Our timings were obtained on a Linux desktop (3.1GHz Intel processor, 8.0Gb total memory) with {\sc Maple} 16 (command line interface), {\sc Mathematica} 9 (graphical interface) and {\sc Qepcad-B} 1.69.  
For each experiment we produce a CAD and give the time taken and number of cells (cell count). The first is an obvious metric while the second is crucial for applications performing operations on each cell.  

For {\sc Qepcad} the options {\tt +N500000000} and {\tt +L200000} were provided, the initialization included in the timings and explicit equational constraints declared when present with the product of those from the individual QFFs declared otherwise.  
In {\sc Mathematica} the output is not a CAD but a formula constructed from one \cite{Strzebonski10}, with the actual CAD not available to the user.  Cell counts for the algorithms were provided by the author of the {\sc Mathematica} code. 

TTICADs are calculated using our implementation described in Section \ref{MapleSec}, which is simple and not optimized.  The results in this section are not presented to claim that our implementation is state of the art, but to demonstrate the power of the TTICAD theory over the the conventional theory, and how it can allow even a simple implementation to compete.  Hence the cell counts are of most interest.

The time is measured to the nearest tenth of a second, with a time out (T/O) set at $5000$ seconds.   When {\bf F} occurs it indicates failure due to a theoretical reason such as not well-oriented (in either sense).  The occurrence  of Err indicates an error in %\texttt{TRDisolaterealroots}, 
an internal subroutine of {\sc Maple}'s \texttt{RegularChains} package, used by \texttt{ProjectionCAD}.  This error is not theoretical but a bug, 
%which should be fixed in the future and is 
beyond our control. 

We considered examples originating from \cite{BuchbergerHong1991}.  However these problems (and most others in the literature) involve conjunctions of conditions, chosen as such to make them amenable to existing technologies.
%such as equational constraints or preconditioning by Gr\"obner Bases. 
These problems can be tackled using TTICAD, but they do not demonstrate its full strength. Hence we introduced new examples, denoted with a $\dagger$, which are adapted from \cite{BuchbergerHong1991} to have disjuncted QFFs.

Two examples came from the application of branch cut analysis for simplification.  These problems require a decomposition according to branch cuts of the form $f = 0 \land g < 0$, and then go on to test the validity of a simplification on each cell, \cite[etc.]{Phisanbutetal2010a}.  We need to consider the disjunction of the branch cuts making such problems suitable for Algorithm \ref{algorithm:TTICADalgorithm}.  We included a key example from Kahan \cite{Kahan1987b}, along with the problem induced by considering the validity of the double angle formulae for arcsin.  Finally we considered the worked example from Section \ref{workedexample} and its generalisation to three dimensions.
Note that A and B following the problem name indicate different variable orderings.  Full details for all examples can all be found in the CAD repository \cite{Wilsonetal2012b}
%.\footnote{available freely online at \texttt{http://opus.bath.ac.uk/29503}
 available freely online at \texttt{http://opus.bath.ac.uk/29503}.
 
\subsection{Results}

We present our results in Table \ref{table:CADW:TTICAD}.  For each problem we give the name used in the repository, $n$ the number of variables, $d$ the maximum degree of polynomials involved and $t$ the number of QFFs used for TTICAD.  We then give the time taken and number of cells produced by each algorithm.  

We first compare our TTICAD implementation with the sign-invariant CAD generated using \texttt{ProjectionCAD} with McCallum's projection operator \cite{ProjectionCAD}.  Since these use the same architecture the comparison makes clear the benefits of the TTICAD theory.  The experiments confirm the fact that  
%since each cell of a TTICAD is a superset of cells from a full sign-invariant CAD and hence the 
the cell count for TTICAD will always be less than or equal to that of a sign-invariant CAD produced using the same implementation.  Ellipse$\dagger$ A is not well-oriented in the sense of \cite{McCallum1998}, and so both methods return {\bf FAIL}.  Solotareff$\dagger$ A and B are well-oriented in this sense but not in the stronger sense of Definition \ref{def:tticadwellorientedness} and hence TTICAD fails while the full sign-invariant CADs can be produced.  The only example with equal cell counts is Collision$\dagger$ A in which the non-equational constraints were so simple that the projection polynomials were unchanged.  Examining the results for the worked example and its generalisation we start to see the true power of TTICAD. In 3D Example A we see a 759-fold reduction in time and a 50-fold reduction in cell count. 

\begin{sidewaystable*}
\caption{Comparing TTICAD to the full CAD built with the same architecture and other CAD algorithms.} %\vskip+3pt
\begin{center}
\label{table:CADW:TTICAD}
\begin{tabular}{lccc|rr|rr|rr|rr|rr}
  \multicolumn{4}{c|}{Problem} 
& \multicolumn{2}{|c|}{Full-CAD} & \multicolumn{2}{|c|}{TTICAD} 
& \multicolumn{2}{|c|}{{\sc Qepcad}} & \multicolumn{2}{|c|}{{\sc Maple}}  
& \multicolumn{2}{|c}{{\sc Mathematica}} \\
Name & n & d & t
                        &  Time      &  Cells      &  Time      &  Cells       
                        &  Time      &  Cells      &  Time      &  Cells       
                        &  Time      &  Cells \\
\hline
Intersection A          & 3 & 2 & 1
						& 360.1     &  3707       &  1.7       &  269         
						& 4.5       &  825        &  ---       &  Err     
						& 0.0 & 3   \\
Intersection B          & 3 & 2 & 1
						& 332.2     &  2985       &  1.5       &  303         
						& 4.5       &  803        &  50.2      &  2795    
						& 0.0 & 3   \\
Random A                & 3 & 3 & 1
						& 268.5     &  2093       &  4.5       &  435         
						& 4.6       &  1667       &  23.0      &  1267    
						& 0.1 & 657 \\
Random B                & 3 & 3 & 1
						& 442.7     &  4097       &  8.1       &  711         
						& 5.4       &  2857       &  48.1      &  1517    
						& 0.0 & 191 \\ 
Intersection$\dagger$ A & 3 & 2 & 2
						&  360.1     &  3707       &  68.7      &  575         
						&  4.8       &  3723       &  ---       &  Err     
						& 0.1 & 601 \\
Intersection$\dagger$ B & 3 & 2 & 2
						&  332.2     &  2985       &  70.0      &  601         
						&  4.7       &  3001       &  50.2      &  2795    
						& 0.1 & 549 \\
Random$\dagger$ A       & 3 & 3 & 2
						&  268.5     &  2093       &  223.4     &  663         
						&  4.6       &  2101       &  23.0      &  1267    
						& 0.2 & 808 \\
Random$\dagger$ B       & 3 & 3 & 2
						&  442.7     &  4097       &  268.4     &  1075        
						&  142.4     &  4105       &  48.1      &  1517    
						& 0.2 & 1156 \\
Ellipse$\dagger$ A      & 5 & 4 & 2
						&  ---       &  {\bf F}    &  ---       &  {\bf F}     
						&  291.6     &  500609     &  1940.1    &  81193   
						& 11.2 & 80111 \\
Ellipse$\dagger$ B      & 5 & 4 & 2
						&  T/O       &  ---        &  T/O       &  ---         
						&  T/O       &  ---        &  T/O       &  ---     
						& 2911.2 & 16603131 \\
Solotareff$\dagger$ A   & 4 & 3 & 2
						&  677.6     &  54037      &  46.1      &  {\bf F}    
						&  4.9       &  20307      &  1014.2    &  54037   
						& 0.1 & 260 \\   
Solotareff$\dagger$ B   & 4 & 3 & 2
						&  2009.2    &  154527     &  123.8     &  {\bf F}     
						&  6.3       &  87469      &  2951.6    &  154527  
						& 0.1 & 762 \\
Collision$\dagger$ A    & 4 & 4 & 2
						&  264.6     &  8387       &  267.7     &  8387        
						&  5.0       &  7813       &  376.4     &  7895    
						& 3.6 & 7171 \\ 
Collision$\dagger$ B    & 4 & 4 & 2
						&  ---       &  Err        &  ---         &  Err       
						& T/O        &  ---        &  T/O       &  ---     
						& 591.5 & 1234601 \\
Kahan A                 & 2 & 4 & 7
						&  10.7      &  409        &  0.3       &  55          
						&  4.8       &  261        &  15.2      &  409     
						& 0.0 & 72 \\       
Kahan B                 & 2 & 4 & 7
						&  87.9      &  1143       &  0.3       &  39        
						&  4.8       &  1143       &  154.9     &  1143    
						& 0.1 & 278 \\
Arcsin A                & 2 & 4 & 4
						&  2.5       &  225        &  0.3       &  57          
						&  4.6       &  225        &  3.3       &  225     
						& 0.0 & 175 \\       
Arcsin B                & 2 & 4 & 4
						&  6.5       &  393        &  0.2       &  25          
						&  4.5       &  393        &  7.8       &  393     
						& 0.0 & 79 \\
2D Example A            & 2 & 2 & 2
						&  5.7       &  317        &  1.2       &  105         
						&  4.7       &  249        &  6.3       &  317     
						& 0.0 & 24 \\
2D Example B            & 2 & 2 & 2
						&  6.1       &  377        &  1.5       &  153         
						&  4.5       &  329        &  7.2       &  377     
						& 0.0 & 175 \\
3D Example A            & 3 & 3 & 2
						&  3795.8    &  5453       &  5.0       &  109         
						&  5.3       &  739        &  ---       &  Err     
						& 0.1 & 44 \\
3D Example B            & 3 & 3 & 2
						&  3404.7    &  6413       &  5.8       &  153         
						&  5.7       &  1009       &  ---       &  Err     
						& 0.1 & 135 \\
\end{tabular}
\end{center}
\vskip-20pt
\end{sidewaystable*}

We next compare our implementation of TTICAD with the state of the art in CAD: {\sc Qepcad} \cite{Brown03}, {\sc Maple} \cite{Chenetal2009d} and {\sc Mathematica} \cite{Strzebonski06, Strzebonski10}.  {\sc Mathematica} is the quickest, however TTICAD often produces fewer cells.  We note that {\sc Mathematica}'s algorithm uses powerful heuristics and so actually used Gr\"obner bases on the first two problems, causing the cell counts to be so low.  When all implementations succeed TTICAD usually produces far fewer cells than {\sc Qepcad} or {\sc Maple}, especially impressive given {\sc Qepcad} is producing partial CADs for the quantified problems, while TTICAD is only working with the polynomials involved.  
For Collision$\dagger$ A the TTICAD theory offers no benefit allowing the better optimized alternatives to have a lower cell count.
 
Reasons for the TTICAD implementation struggling to compete on speed in general are that the {\sc Mathematica} and {\sc Qepcad} algorithms are largely implemented directly in {\tt C}, have had far more optimization, and in the case of {\sc Mathematica} use validated numerics for lifting \cite{Strzebonski06}.  However, the strong performance in cell counts is very encouraging, both due its importance for applications where CAD is part of a wider algorithm (such as branch cut analysis) and for the potential if TTICAD theory were implemented elsewhere.

\section{Conclusions}
\label{sec:Conc}

We have defined Truth Table Invariant CADs, which can be more closely aligned to the needs of problems than traditional sign-invariant CADs.  Theorem \ref{DJW:theorem:SMcC3corrected} extended the theory of equational constraints allowing us to develop Algorithm \ref{algorithm:TTICADalgorithm} to construct TTICADs efficiently for a large range of problems.  The algorithm has been implemented in {\sc Maple} giving promising experimental results.  TTICADs in general have less cells than full sign-invariant CADs using the same implementation and we showed that this allows even a simple implementation of TTICAD to compete with the state of the art CAD implementations.  %, especially for problems which cannot be tackled by equational constraints alone.  
It is anticipated that future implementations of TTICAD could be far better optimized leading to lower times for the same cell counts.  We also note that the benefits of TTICAD increase with the number of QFFs in a problem and so larger problems may be susceptible to TTICAD when other approaches fail.
 
We hope that these results inspire other implementations of TTICAD, with Corollary \ref{corollary:rescad} showing a particularly easy way to adapt existing CAD implementations.    

\subsection{Future Work}

There is scope for optimizing the algorithm and extending it to allow less restrictive input.  Lemma \ref{lemma:constpolys} gives one extension that is included in our implementation while other possibilities include removing some of the caution implied by well-orientedness, analogous to \cite{Brown2005}.  Also, work developing heuristics for composing the input is underway in \cite{BDEW13}.  

Of course, the implementation of TTICAD used here could be improved in many ways, but perhaps more desirable would be for TTICAD to be incorporated into existing state of the art CAD implementations.  In particular, we would like to use the existing {\sc Maple} CAD commands \cite{Chenetal2009d} but this requires first understanding when they give order-invariance, a key question currently under consideration.  
We see several possibilities for the theoretical development of TTICAD:
\begin{itemize}[itemsep=-4.5pt,topsep=-10pt]
\item Can we apply the theory recursively instead of only at the top level?  For example by widening the projection operator to 
%allow enough information to 
conclude order-invariance, as in \cite{McCallum2001}.  
\item Can we apply TTICAD to forms of QFF other than ``one equality and other items''?
For example, can we generalise the theory of bi-equational constraints?
% to apply within QFFs?
\item Can we make use of the ideas behind partial CAD to avoid unnecessary lifting once the truth value of a QFF on a cell is determined? 
\item Can anything be done when $\Phi$ is not well oriented?
\item Can we implement the lifting algorithm in parallel?
\end{itemize}

\subsection*{Acknowledgements}
We are grateful to A.~Strzebo\'nski for assistance in performing the Mathematica tests and to the anonymous referees for useful comments.  We also thank the rest of the Triangular Sets seminar at Bath (A.~Locatelli, G.~Sankaran and N.~Vorobjov) for their input, and the team at Western University (C.~Chen, M.~Moreno Maza, R.~Xiao and Y.~Xie) for access to their {\sc Maple} code and helpful discussions. This work was supported by the EPSRC grant: EP/J003247/1.

\bibliography{jhd-edit}

\end{document}